\documentclass[11pt,a4paper]{amsart}
\usepackage{times,epsfig}
\usepackage{amsfonts,amscd,amsmath,amssymb,amsthm}

\numberwithin{equation}{section}

\DeclareMathOperator{\sign}{sign}

\newcommand{\R}{\mathbb{R}}

\newcommand{\C}{\mathbb{C}}

\newcommand{\1}{\mathbb{I}}

\newcommand{\cA}{{\mathcal A}}
\newcommand{\cB}{{\mathcal B}}

\newcommand{\cE}{{\mathcal E}}

\newcommand{\cG}{{\mathcal G}}

\newcommand{\cI}{{\mathcal I}}

\newcommand{\cK}{{\mathcal K}}
\newcommand{\cL}{{\mathcal L}}

\newcommand{\cT}{{\mathcal T}}

\newcommand{\cV}{{\mathcal V}}

\newcommand{\diag}{\mathrm{diag}}

\newcommand{\ii}{\mathrm{i}}
\newcommand{\e}{\mathrm{e}}

{\bf}{\it}
{\bf}{\it}
{\bf}{\it}
{\bf}{\it}
{\bf}{\it}

\newtheorem{theorem}{Theorem}[section]{\bf}{\it}
\newtheorem{proposition}[theorem]{Proposition}{\bf}{\it}
\newtheorem{corollary}[theorem]{Corollary}{\bf}{\it}
{\it}{\rm}
\newtheorem{lemma}[theorem]{Lemma}{\bf}{\it}
\newtheorem{remark}[theorem]{Remark}{\it}{\rm}
{\bf}{\it}
{\bf}{\it}
{\bf}{\it}
{\bf}{\it}

\title{Grassmann--Gaussian integrals and generalized star products\\
\vspace{0.3cm}
{\Small \emph{In memory of Al.B. Zamolodichikov}}
}

\author[Sh. Khachatryan] {Shahane Khachatryan}
\address{Shahane Khachatryan\\ Yerevan Physics Institute, Br. Alikhanyan 2\\ Yerevan 36, Armenia}
\email{shah@yerphi.am}

\author[R. Schrader]{Robert Schrader}

\address{Robert Schrader\\ Institut f\"{u}r
Theoretische Physik\\ Freie Universit\"{a}t Berlin, Arnimallee 14
\\ D-14195 Berlin,
Germany}
\email{schrader@physik.fu-berlin.de}

\author[A. Sedrakyan]{ Ara Sedrakyan}
\address{Ara Sedrakyan\\ Yerevan Physics Institute, Br. Alikhanyan 2\\ Yerevan 36, Armenia}
\email{sedrak@nbi.dk}

\thanks{$^\ast$ R.S. and A.S. are supported in part by the Humboldt Foundation}



\begin{document}

\begin{abstract} In quantum scattering on networks there is a non-linear composition rule for on-shell 
scattering matrices which serves as a replacement for the multiplicative rule of transfer matrices valid 
in other physical contexts. In this article, we show how this composition rule is obtained using 
Berezin integration theory with Grassmann variables.
\end{abstract}

\maketitle

\section{Introduction}\label{sec:1}

   Potential scattering for one particle Schr\"{o}dinger operators on the line possesses 
a remarkable property concerning its (on-shell) scattering matrix given as a 
2 × 2 matrix-valued function of the energy. Let the potential $V$ be given as the sum of 
two potentials $V_1$ and $V_2$ with disjoint support. Then the scattering matrix for $V$ at a 
given energy is obtained from the two scattering matrices for $V_1$ and $V_2$ at 
the same energy by a certain non-linear, noncommutative but associative composition rule. 
This fact in quantum scattering theory on the line has been discovered independently by several
authors (see, e.g., \cite{KSEI,Aktosun,RRT,SV1,SV2} and is an easy consequence of the 
multiplicative property 
of the transfer matrix of the Schr\"{o}dinger equation (see e.g. \cite{KS3}). It has also
been well known in the theory of mesoscopic systems and multichannel conductors (see, e.g.,
\cite{Tong,Dorokhov1,Dorokhov2,Dorokhov3, Dorokhov4,MPK,SMMP,Beennakker,Datta, ET}). In higher space 
dimensions a similar rule is unlikely to exist due to the defocusing of wave packets under propagation.
However, for large separation between the supports of the potentials the scattering matrix at a
given energy may asymptotically be expressed in terms of the scattering matrices for $V_1$ and $V_2$ at 
the same energy \cite{KS2,KS3}.

   To the best of our knowledge the composition rule for $2\times 2$ scattering matrices was first
observed in the context of electric network theory by Redheffer \cite{Redheffer1,Redheffer2}, 
who called it the \emph{star product}. Now there are situations, where 
the concept of a transfer matrix cannot always be introduced but where the (on shell) scattering matrix 
nevertheless exists. An important example is given by quantum dynamical models on graphs, that is quasi-one 
dimensional quantum systems, and which are described by Schr\"{o}dinger operators. 
Such systems are nowadays a subject of intensive study (see, e.g., \cite{Kuchment,BCFK, EKKST} and 
references quoted there). In the article \cite{KS4} a composition rule, called the 
\emph{generalized star product}, was introduced and further analyzed in \cite{KS5}. This composition 
rule extends the star product of Redheffer and allows one to obtain the on-shell scattering matrix on a 
given graph from the on-shell scattering matrices associated with sub-graphs. 
The generalized star product is 
defined for arbitrary matrices but for unitary matrices the outcome is also unitary.

In this article we provide a new way of obtaining the generalized star product. The method is based on 
the Grassmannian (fermionic) integration theory given by Berezin, \cite{Berezin} and it evaluates certain 
Gauss -- Grassmann integrals. 
In addition we also show, how one can arrive at the generalized star product using ordinary 
Gaussian (bosonic) distributions. Then, however, one has to work with a restriction, the covariances have 
to be positive definite matrices.

This technique of using Gaussian integration with fermionic fields permits an action formulation of 
some network models. Thus for example it can be applied to the Chalker -- Coddington network model 
\cite{CC} to describe plateau transitions in the quantum Hall effect \cite{Sedrak}. The method is also very 
convenient for an investigation of models with a large number of scattering centers. In this limit, as well 
as at the critical point, one can give an equivalent quantum field theoretic formulation of network models 
\cite{AS2}.  

The article is organized as follows. In the next section in order to establish notation and for the 
convenience of the reader we briefly recall the basic notions of Berezin's theory. Though most of the 
material can be found in standard text books of quantum field theory, see, e.g., \cite{IZ, Weinberg}, 
the relations we need seem not to be so easily accessible.
In Section \ref{sec:3} we show how to obtain the generalized star product using Gauss -- Grassmann 
integrals.
In the Appendix we briefly discuss the corresponding bosonic version, that is how the generalized 
star product can be obtained from the standard theory of Gaussian distributions.

\section{Preliminaries}\label{sec:2}

In this section we first briefly review Grassmann integration and then we recall some concepts from 
graph theory that we will need.

\vspace{0.3cm}
\subsection{ Grassmann Integration }

In this subsection we briefly recall the basic notions concerning Grassmann variables and the 
associated integration theory, see \cite{Berezin}, and which we will need.
Let $\overline{a}_i,a_i$ be Grassmann variables, which means they anti-commute  
$$
\overline{a}_i \overline{a}_j=-\overline{a}_j\overline{a}_i,\quad 
\overline{a}_i a_j=-a_j\overline{a}_i,\quad a_i a_j=-a_j a_i.
$$
We denote the associated (complex) Grassmann algebra with unit $\1$ by $\cA_I$. 
Elements $\alpha$ in $\cA_I$ have a unique representation in the form 
\begin{equation}\label{gr1}
\alpha=\alpha(\overline{a},a)=\sum_{J,K\subseteq I}c_{J,K}\,\overline{a}_J\, a_K,\quad c_{J,K}\in\C 
\end{equation}
with the anti-ordered and ordered products
$$
\overline{a}_J=\buildrel\curvearrowleft\over\prod_{j\in J}\overline{a}_j,
\quad a_K=\buildrel\curvearrowright\over\prod_{j\in K}a_j,
\qquad J\neq \emptyset,K\neq \emptyset
$$
and $\overline{a}_\emptyset=a_\emptyset=\1$. By definition 
\begin{equation}\label{gr2}
\alpha(\overline{a}=0,a=0)=c_{J=\emptyset,K=\emptyset}.
\end{equation}

Correspondingly the subalgebra generated by the elements $\overline{a}_j,a_j$ with 
$i\in J$ will be denoted by $\cA_J$.
In addition we introduce symbols 
$d\overline{a}_i$ and $da_i$ which anti-commute among themselves and with the $\overline{a}_i,a_i$
and define the anti-ordered and ordered products 
$$
d\overline{a}_J=\buildrel\curvearrowleft\over\prod_{j\in J}d\overline{a}_j,
\qquad  da_K=\buildrel\curvearrowright\over\prod_{j\in K}da_j.
$$
(Berezin-)Integration over $L\subseteq I $ is defined as 
\begin{equation}\label{gr3}
\int\overline{a}_J \,a_K \,d\overline{a}_{L}\,da_{L}=
\begin{cases}\quad0&\quad L\not\subseteq J\quad\mbox{or}\quad L\not\subseteq K\\
&\\
\quad(-1)^{\mid L\mid}\sign^L_{JK}\: \overline{a}_{J\setminus L}\,a_{K\setminus L}&
\quad L\subseteq J\quad\mbox{and}\quad L\subseteq K  
\end{cases}
\end{equation}
where $\sign^L_{JK}=\pm 1$ is defined when $L\subseteq J$ and $L\subseteq K$ and is such that 
$$
\overline{a}_J \,a_K=\sign^L_{JK}\:\overline{a}_{J\setminus L}\,a_{K\setminus L}\,\overline{a}_L\,a_L
$$ 
holds. We call $d\overline{a}_{L}\,da_{L}$ the \emph{volume form} associated to the index set $L$. 
From the definition \eqref{gr3} it is easy to see that integration may be performed in 
steps (Fubini's Theorem), an observation which will become important in our discussion. 
As a special case of \eqref{gr3} and with $\alpha\in\cA_I$ written in the form \eqref{gr1}
\begin{equation}\label{gg3}
\int\,\alpha\,d\overline{a}_{I}\, da_{I}=(-1)^{\mid I\mid}c_{I,I}.
\end{equation}
In analogy to the Lebesgue integral over $\R^n$ this integral exhibits a \emph{translation invariance}
in the following form. Introduce additional Grassmann variables 
$\overline{b}_i,b_i$, again with the index $i$ in $I$ and which anti-commute with all the previous 
Grassmann variables. 
With $\alpha(\overline{a},a)$ as in \eqref{gr1}, by 
$\alpha(\overline{a}-\overline{b},a-b)$ we understand the expression   
\begin{equation}\label{gr4}
\alpha(\overline{a}-\overline{b},a-b)=\sum_{J,K\subseteq I}c_{J,K}\,(\overline{a}-\overline{b})_J\, (a-b)_K 
\end{equation}
with the anti-ordered and ordered products
$$
(\overline{a}-\overline{b})_J=\buildrel\curvearrowleft\over\prod_{j\in J}(\overline{a}_j-\overline{b}_j),
\quad (a-b)_K=\buildrel\curvearrowright\over\prod_{j\in K}(a-b)_j.
$$
Expanding each $(\overline{a}-\overline{b})_J\, (a-b)_K $ into a sum of monomials 
$\overline{a}_{J^\prime}\,a_{K^\prime}$ it is easy to establish translation invariance 
of the integral in the form
$$ 
\int \alpha(\overline{a}-\overline{b},a-b)\, d\overline{a}_{I}\,da_{I}=
\int\,\alpha(\overline{a},a)\, d\overline{a}_{I}\,da_{I}
$$
or even more generally 
\begin{equation}\label{gr5}
\int \alpha(\overline{a}-\overline{b},a-b;\overline{b},b,\overline{c},c)\, d\overline{a}_{I}\,da_{I}=
\int\,\alpha(\overline{a},a;\overline{b},b,\overline{c},c)\, d\overline{a}_{I}\,da_{I},
\end{equation}
valid as a relation in $\cB$, the Grassmann algebra generated by the $\overline{b}_i$ and $b_i$ and 
possibly additional Grassmann variables $\overline{c}_k$ and $c_k$. 
This seemingly trivial relation will also become very useful below.
Let $A$ be any complex $n\times n$ matrix and choose 
$I_n=\{1,2,\cdots,n\}$ to be the index set. Set  
$$
\overline{a}\cdot Aa=\sum_{1\le i,j \le n}\,\overline{a}_i A_{ij}a_j. 
$$ 
The the Gauss -- Grassmann integral can be calculated 
\begin{equation}\label{gr7}
\int \e^{-\overline{a}\cdot Aa}\,d\overline{a}_{I_n}\,da_{I_n}=\det A,
\end{equation}
the Gaussian distribution analogue of which is relation\eqref{gauss02} in the Appendix. 
In a further analogy to ordinary Gaussian distributions, see \eqref{gauss6} below, 
we obtain 
\begin{equation}\label{gr11}
\begin{aligned} 
\int\, \overline{a}_i\, \e^{-\overline{a}\cdot Aa}\,d\overline{a}_{I_n}\,da_{I_n}&=
\int\,a_i\, \e^{-\overline{a}\cdot Aa}\,d\overline{a}_{I_n}\,da_{I_n}=0\\
\int\, a_i\,\overline{a}_j\, \e^{-\overline{a}\cdot Aa}\,d\overline{a}_{I_n}\,da_{I_n}
&=\det A\:\:A^{-1}_{ij}.
\end{aligned}
\end{equation}
We need an extension of \eqref{gr7} when $\det A\neq 0$. Set 
$$
\overline{b}\cdot a=-a\cdot\overline{b}=\sum_{1\le i\le n}\overline{b}_i a_i ,
\quad \overline{a}\cdot b=-b\cdot\overline{a}= 
\sum_{1\le i\le n}\overline{a}_i b_i.
$$
Then one can prove 
\begin{equation}\label{gr10} 
\int \e^{-\overline{a}\cdot Aa +\overline{b}\cdot a+\overline{a}\cdot b}\,d\overline{a}_{I_n}\,da_{I_n}=\det A \:\:
\e^{\overline{b}\cdot A^{-1}b}
\end{equation}
by using the translation invariance \eqref{gr5} of the integral and by completing the square in the 
exponent.
Finally we provide a variant of \eqref{gr10}, which will become important below. 
For any $1\le p\le n$ consider $I_p$ as a subset of $I_n$ and let $I^c_p=\{p+1,\cdots,n\}$ be its 
complement. Let $\overline{a}^{(1)},a^{(1)}$ denote the set of Grassmann variables $\overline{a}_i,a_1$ 
with $1\le i\le p$ and $\overline{a}^{(2)},a^{(2)}$ those with $p+1\le i\le n$.
For a matrix $A$ as before, consider the corresponding matrix block decomposition  
\begin{equation}\label{Ablock}
A=\begin{pmatrix}A_{11}&A_{12}\\A_{21}&A_{22}\end{pmatrix}.
\end{equation}
where $A_{11}$ is a $p\times p$ matrix, $A_{12}$ a $p\times(n-p$ matrix etc. and set 
\begin{align}\label{gr110}
\overline{a}^{(1)}\cdot A_{11}a^{(1)}&=\sum_{1\le i,j\le p}\overline{a}_iA_{ij}a_j,\qquad\qquad
\overline{a}^{(1)}\cdot A_{12}a^{(2)}=\sum_{1\le i\le p, p+1\le j\le n}\overline{a}_iA_{ij}a_j \\\nonumber 
\overline{a}^{(2)}\cdot A_{21}a^{(1)}&=\sum_{p+1\le i\le n,1\le j\le p}\overline{a}_iA_{ij}a_j,\quad
\overline{a}^{(2)}\cdot A_{22}a^{(2)}=\sum_{p+1\le i,j\le n}\overline{a}_iA_{ij}a_j\\\nonumber
\end{align}
such that 
\begin{equation}\label{gr111}
\overline{a}\cdot Aa=\overline{a}^{(1)}\cdot A_{11}a^{(1)}+\overline{a}^{(1)}\cdot A_{12}a^{(2)}
+\overline{a}^{(2)}\cdot A_{21}a^{(1)}+\overline{a}^{(2)}\cdot A_{22}a^{(2)}
\end{equation}
holds.
\begin{lemma}\label{lem:a}
Let the $(n-p)\times (n-p)$ matrix $A_{22}$ be invertible and set 
$$
\widehat{A}=A_{11}-A_{12}A_{22}^{-1}A_{21},
$$ 
a $p\times p$ matrix. Then 
\begin{equation}\label{gr112}
\int \e^{-\overline{a}\cdot Aa}\,d\overline{a}_{I_n^c}\,da_{I_n^c}=\det A_{22} \cdot\:
\e^{-\overline{a}^{(1)}\cdot \widehat{A}a^{(1)}}
\end{equation}
holds. Therefore the relation
\begin{equation}\label{gr113}
\det A=\det A_{22} \cdot\det \widehat{A} 
\end{equation}
is valid. If in addition $\widehat{A}$ is invertible, the matrix elements of its inverse are the 
corresponding ones for $A^{-1}$
\begin{equation}\label{gr114}
\widehat{A}^{-1}_{\quad ij}=A^{-1}_{\quad ij}\;,\qquad 1\le i,j\le p.
\end{equation}
\end{lemma}
\begin{proof}
We use the decomposition \eqref{gr111} and complete the square to obtain
$$
\e^{-\overline{a}\cdot Aa}=\e^{-(\overline{a}^{(2)}+A_{22}^{-1\,T}A_{21}^{T}\overline{a}^{(1)})
\cdot A_{22}(a^{(2)}+A_{22}^{-1}A_{21}a^{(1)})}\;\e^{-\overline{a}^{(1)}\cdot \widehat{A}a^{(1)}}
$$
and \eqref{gr112} follows from \eqref{gr7} and translation invariance. \eqref{gr113} in turn follows from
\eqref{gr7} by integrating \eqref{gr112} out over the Grassmann variables $\overline{a}_i,a_i$ with 
$i\in I_p$. As for the last claim, if $\widehat{A}$ is invertible, then so is $A$ by \eqref{gr113} 
(and conversely under the present assumption that $A_{22}$ is invertible).
By the last relation in \eqref{gr11} and by \eqref{gr112} for all $ 1\le i,j\le p$
\begin{align*}
\det \widehat{A}\cdot \widehat{A}^{-1}\;_{ ij}&
=\int a_i\overline{a}_j\e^{-\overline{a}\cdot\widehat{A} a}\,d\overline{a}_{I_p}\,da_{I_p}\\
&=\frac{1}{\det A_{22}}\int a_i\overline{a}_j\e^{-\overline{a}\cdot A  a}\,d\overline{a}_{I_n}\,da_{I_n}
=\frac{\det A}{\det A_{22}}\cdot A^{-1}_{\quad ij}
\end{align*}
and \eqref{gr114} follows from \eqref{gr113}.
\end{proof}
\begin{remark}\label{re:1}
If $A_{22}$ is not invertible, the left hand side of \eqref{gr112} is still well defined but not the 
right hand side. Relation \eqref{gr113} also follows from the factorization of a block matrix. 
It involves the \emph{Schur complement} of $A_{22}$, which is just $\widehat{A}$, 
see, e.g., \cite{GMW,Zhang}. In addition the inverse of the Schur complement enters the inverse of $A$ as 
one block part and this is just relation \eqref{gr114}
\begin{align}\label{schur}
\begin{pmatrix}A_{11}&A_{12}\\
A_{21}&A_{22}\end{pmatrix}^{-1}=\begin{pmatrix}\widehat{A}^{-1}&
-\widehat{A}^{-1}A_{12}A_{22}^{-1}\\-A_{22}^{-1}A_{21}\widehat{A}^{-1}&A_{22}^{-1}
+A_{22}^{-1}A_{21}\widehat{A}^{-1}A_{12}A_{22}^{-1}\end{pmatrix}.
\end{align}
\end{remark}

\subsection{Some basic concepts from graph theory}
We first recall some notions, which will be useful in the sequel.
A finite noncompact graph is a 4-tuple $\cG=(\cV,\cI,\cE,\partial)$, where
$\cV$ is a finite set of \emph{vertices}, $\cI$ is a finite set of
\emph{internal edges}, $\cE$ is a finite set of \emph{external edges}. Set 
$$
n(\cV)=\mid\cV\mid,\quad n(\cI)=\mid\cI\mid,\quad  n(\cE)=\mid\cE\mid,
$$
the number of elements in these sets. We assume each of these sets $\cV,\cE,\cI$ 
to be ordered in some arbitrary but 
fixed way. This induces an ordering in $\cE\cup\cI$, where by definition elements in $\cE$ come first. 
On the product set $(\cE\cup\cI)\times \cV$ by definition the induced ordering $\preceq$ is such 
$(i,v)\preceq (i^\prime,v^\prime)$ if and only if $v\prec v^\prime$ or $v=v^\prime$ and $i\preceq i^\prime$.
Elements in $\cI\cup\cE$ are called \emph{edges}. The map $\partial$ assigns
to each internal edge $i\in\cI$ an ordered pair of (possibly equal) vertices
$\partial(i):=\{v_1,v_2\}$ and to each external edge $e\in\cE$ a single
vertex $v$. The vertices $v_1=:\partial^-(i)$ and $v_2=:\partial^+(i)$ are
called the \emph{initial} and \emph{final} vertex of the internal edge
$i$, respectively. The vertex $v=\partial(e)$ is the initial vertex of the
external edge $e$. If $\partial(i)=\{v,v\}$, that is,
$\partial^-(i)=\partial^+(i)$ then $i$ is called a \emph{tadpole}. 
However, to facilitate our exposition, we will exclude tadpoles. 
Two vertices $v$ and $v^\prime$ are \emph{adjacent}, if there is $i\in\cI$ such that these 
vertices form $\partial(v)$.
To any $v\in \cV$ we associate the set of edges terminating at $v$, 
$\cI(v)=\{i\in \cE\cup\cI\mid v\in \partial(i)\}$. We set $n(v)=\mid\cI(v)\mid$, 
the number of edges terminating at $v$. Each of these sets inherits an ordering from the ordering  
of $\cE\cup\cI$. Also 
$$
n(v,v^\prime)= n(v^\prime,v)=\mid \cI(v)\cap \cI(v^\prime)\mid, \quad v\neq v^\prime 
$$
is the number of (internal) edges connecting $v$ with $v^\prime$, so 
$n(v,v^\prime)\le \min (n(v),n(v^\prime))$. $n(v,v^\prime)$ is called the \emph{connectivity matrix}
of the graph $\cG$. $\cI(v)\cap \cI(v^\prime)\subseteq \cI$ holds for 
$v\neq v^\prime $ and likewise this set inherits an ordering from the ordering of $\cI$.
As (unordered) sets each $\cI(v)$ is a disjoint union
$$
\cI(v)=(\cI(v)\cap\cE)\cup_{v^\prime:\,v^\prime\neq v}(\cI(v)\cap\cI(v^\prime)).
$$
By definition, a graph is \emph{compact} if $\cE=\emptyset$, otherwise it is
\emph{noncompact}. Throughout the whole work we will assume that the graph $\cG$ 
is connected, that is, for any $v,v^\prime\in V$ there is an ordered sequence $\{v_1=v,
v_2,\ldots, v_n=v^\prime\}$ such that any two successive vertices in this
sequence are adjacent. In particular, this implies that any vertex of the
graph $\cG$ has nonzero degree, i.e., for any vertex $v$ there is at least one
edge with which it is incident, $n(v)>0$. In addition $n(\cI)\ge n(\cV)-1$ is valid.
By definition a \emph{star graph} is a connected graph which has no internal edges, only one vertex, 
and at least one external edge. 

A graph can be equipped as follows with a metric structure. To any internal
edge $i\in\cI$ we associate an interval $[0,a_i]$ with $a_i>0$
such that the initial vertex of $i$ corresponds to $x=0$ and the terminal
one - to $x=a_i$. Any external edge $e\in\cE$ will be associated with a
semi-line $[0,\infty)$ such that $\partial(e)$ corresponds to $x=0$. 
We call the number $a_i$ the \emph{length} of the internal
edge $i$. The set of lengths $\{a_i\}_{i\in\cI}$, which will also be treated
as an element of $\R^{|\cI|}$, will be denoted by $\underline{a}$. A
compact or noncompact graph $\cG$ endowed with a metric structure is called
a \emph{metric graph} $(\cG,\underline{a})$.


\section{The Fermionic  Construction}\label{sec:3}
Given a graph $\cG$, we introduce the following data. To each vertex $v$ we associate a complex 
$n(v)\times n(v)$ matrix $X(v)$ indexed by the set $\cI(v)$. We call $X(v)$ a \emph{vertex matrix}. 
In addition complex, invertible  $n(v,v^\prime)\times n(v,v^\prime)$ matrices $V(v,v^\prime)$ 
 are given for any pair $v\neq v^\prime$. The associated index sets are $\cI(v)\cap\cI(v^\prime)$ 
and they are supposed to satisfy 
\begin{equation}\label{symmetry}
V(v,v^\prime)=V(v^\prime,v)^{-1}. 
\end{equation}
$V(v,v^\prime)$ is called a \emph{connecting matrix} between $v$ and $v^\prime$. 
We write $\underline{X}_\cG=\{X(v)\}_{v\in\cV}$ and $\underline{V}_\cG=\{V(v,v^\prime)\}_{v\neq v^\prime\in\cV}$ 
for these two sets of data. Figure \ref{fig:2} gives a pictorial description for a graph with three vertices.
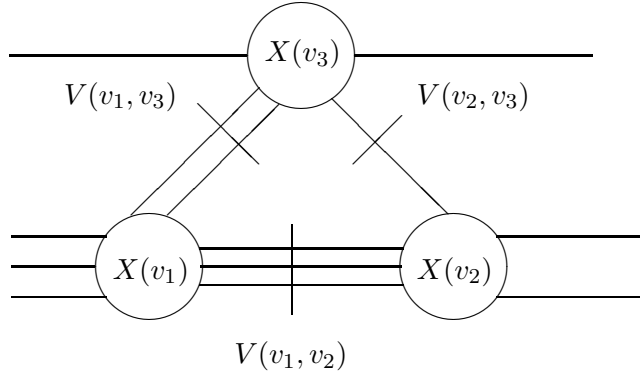
\begin{figure}[htb]
\setlength{\unitlength}{0.8cm}
\begin{picture}(12,8)

\put(1,3){\circle{6.0}}
\put(6,3){\circle{6.0}}
\put(3.5,6.5){\circle{6.0}}

\put(-1.3,6.5){\line(1,0){3,9}}
\put(4.38,6.5){\line(1,0){3,9}}

\put(1.85,3){\line(1,0){3,3}}
\put(1.83,3.3){\line(1,0){3,35}}

\put(1.83,2.7){\line(1,0){3,35}}

\put(-1.26,3.5){\line(1,0){1,56}}
\put(-1.26,3){\line(1,0){1,35}}
\put(-1.26,2.5){\line(1,0){1,56}}

\put(6.7,3.5){\line(1,0){2,35}}
\put(6.7,2.5){\line(1,0){2,35}}

\put(0.7,3.85){\line(1,1){2,11}}
\put(1.3,3.85){\line(1,1){1,83}}
\put(4.0,5.78){\line(1,-1){1,91}}

\put(1.8,5.7){\line(1,-1){1,0}}
\put(-0.4,5.7){$V(v_1,v_3)$}

\put(3.35,3.7){\line(0,-1){1,5}}
\put(2.4,1.4){$V(v_1,v_2)$}

\put(4.35,4.7){\line(1,1){0,8}}
\put(5.4,5.7){$V(v_2,v_3)$}

\put(0.4,2.8){$X(v_1)$}
\put(5.4,2.8){$X(v_2)$}
\put(2.9,6.4){$X(v_3)$}

\end{picture}

\caption{\label{fig:2} A graph with three vertices, $7$ external and 
$6$ internal edges. $X(v_1)$ is a $8\times 8$, $X(v_2)$ a $6\times 6$ and  $X(v_3)$ a $5\times 5$ matrix.
$V(v_1,v_2)$ is a invertible $3\times 3$ matrix, $V(v_2,v_3)$ is just a non-vanishing complex number, while 
$V(v_1,v_3)$ is an invertible $2\times 2$ matrix.} 
\end{figure}
\begin{remark}\label{re:scattering}
Within the context of scattering on quantum graphs, see \cite{KS4,KS5} 
$X(v)$ is the unitary scattering matrix $S(v;E)$ at energy $E>0$ on a single vertex graph with $n(v)$ 
external edges and $V(v,v^\prime)$ is a unitary diagonal matrix 
\begin{equation}\label{metricconn}
V(v,v^\prime)=\diag\left(\e^{\ii \sqrt{E}a_i}\right)_{i\in\cI(v)\cap \cI(v^\prime)}.
\end{equation}
\end{remark}
We also introduce a total of $2(\mid\cE\mid +2\mid\cI\mid)$ anti-commuting Grassmann variables
\begin{equation}\label{grassmannset}
\begin{aligned}
\overline{\eta}_{e,v},\, \eta_{e,v},\, d\overline{\eta}_{e,v},\, d\eta_{e,v},\quad & e \in \cE,\,
v=\partial(e)\\
\overline{\eta}_{i,v},\, \eta_{i,v},\,d\overline{\eta}_{i,v},\, d\eta_{i,v},\quad & i\in\cI,v\in\partial(i)
\end{aligned} 
\end{equation}
and the associated volume forms given as the anti-ordered and ordered products
\begin{equation}\label{volume}
\begin{aligned}
d\overline{\eta}_\cE d\eta_\cE&=\buildrel\curvearrowleft\over\prod_{e\in\cE,v=\partial(e)}
d\overline{\eta}_{e,v}
\buildrel\curvearrowright\over\prod_{e\in\cE,v=\partial(e)}d\eta_{e,v}\\
d\overline{\eta}_\cI d\eta_\cI&=\buildrel\curvearrowleft\over\prod_{i\in\cI,v\in\partial(i)}
d\overline{\eta}_{i,v}
\buildrel\curvearrowright\over\prod_{i\in\cI,v\in\partial(i)}d\eta_{i,v}\\
&\\
d\overline{\eta}_{\cE\cup\cI} d\eta_{\cE\cup\cI}&
=d\overline{\eta}_\cE d\eta_\cE d\overline{\eta}_\cI d\eta_\cI
=d\overline{\eta}_\cE d\overline{\eta}_\cI d\eta_\cE d\eta_\cI=
d\overline{\eta}_\cE d\overline{\eta}_\cI d\eta_\cI \,d\eta_\cE. 
\end{aligned}
\end{equation}
Recall that $\cI(v)\cap \cE$ may be non-empty, so in order to have a compact notation 
we have added the index $v=\partial(e)$ in the definition of $\overline{\eta}_{e,v}$. On the other hand 
$\cI(v)\cap\cI(v^\prime)\cap \cE$ is always empty. For further reference we write 
\begin{equation}
(\cE\cup\cI)\triangleright \cV=\{i,v\}_{i\in\cI(v),v\in\cV}\subset (\cE\cup\cI)\times \cV
\end{equation}
for this index set and with the ordering induced by that of $(\cE\cup\cI)\times \cV$. Set 
\begin{equation}\label{gr12}
\begin{aligned}
\overline{\eta}\cdot \cL(v)\eta&=\sum_{i,j\in\cI(v)}\overline{\eta}_{i,v}X(v)_{ij}\eta_{j,v}\\
\overline{\eta}\cdot \cL(v,v^\prime)\eta&
=-\sum_{i,j\in\cI(v)\cap\cI(v^\prime)}\overline{\eta}_{i,v}V(v,v^\prime)_{ij}\eta_{j,v^\prime}.
\end{aligned}
\end{equation}
with the convention that $\overline{\eta}\cdot \cL(v,v^\prime)\eta=0$ if 
$\cI(v)\cap\cI(v^\prime)=\emptyset$.
Define the \emph{quadratic, fermionic Lagrangian} as
\begin{equation}\label{gr13}
\overline{\eta}\cdot \cL_\cG\, \eta=\sum_{v}\overline{\eta}\cdot \cL(v)\,\eta
+\sum_{v\neq v^\prime}\overline{\eta}\cdot \cL(v,v^\prime)\,\eta.
\end{equation}
We decompose the set of Grassmann variables into \emph{exterior} and \emph{interior variables} 
$$
\begin{aligned}
\overline{\eta}_{\cE}&=\{\overline{\eta}_{e,v=\partial(e)}\}_{e\in\cE},\quad 
\eta_{\cE}=\{\eta_{e,v=\partial(e)}\}_{e\in\cE}\\
\overline{\eta}_{\cI}&=\{\overline{\eta}_{i,v}\}_{i\in\cI,v\in\partial(i)},\quad 
\eta_{\cI}=\{\eta_{i,v}\}_{i\in\cI,v\in \partial(i)}.
\end{aligned}
$$
and correspondingly we set 
\begin{align}\label{gr14}
\overline{\eta}\cdot \cL_\cG\, \eta&=\overline{\eta}_\cE\cdot \cL_\cE\, \eta_\cE+
\overline{\eta}_\cE\cdot \cL_{\cE,\cI}\, \eta_\cI+\overline{\eta}_\cI\cdot \cL_{\cI,\cE}\, \eta_{\cE}+
\overline{\eta}_\cI\cdot \cL_\cI\, \eta_\cI\\\nonumber
&=(\overline{\eta}_\cI+(\cL_\cI^{-1})^T(\cL_{\cE,\cI})^T\, \overline{\eta}_\cE)\cdot \cL_\cI\,
(\eta_\cI+\cL_\cI^{-1}\cL_{\cI,\cE}\, \eta_\cE)\\\nonumber
&\qquad \qquad+\overline{\eta}_\cE\cdot
\left(\cL_\cE-\cL_{\cE,\cI}\cL_\cI^{-1}\cL_{\cI,\cE}\right)\eta_\cE.\nonumber
\end{align}
In analogy to \eqref{Ablock} the first line just corresponds to the following block matrix representation
\begin{equation}\label{gr15}
\cL_\cG=\left(\begin{array}{cc}\cL_\cE&\cL_{\cE,\cI}\\
\cL_{\cI,\cE}&\cL_{\cI}
\end{array}\right)
\end{equation} 
up to a different index ordering. Provided the matrix $\cL_\cI$ is invertible, we may define
its Schur complement
\begin{equation}\label{gr16}
\cK_\cG=\cL_\cE-\cL_{\cE,\cI}\cL_\cI^{-1}\cL_{\cI,\cE}.
\end{equation}
It will also be convenient to introduce the notation 
\begin{equation}\label{gr161}
\cT_\cG=\cL_{\cI}
\end{equation} 
in order to indicate the dependence on $\cG$.
\begin{theorem}\label{grassmann:theo}
If $\cL_\cI$ is invertible, then 
\begin{equation}\label{gr17}
\int\; \e^{-\;\overline{\eta}\cdot\cL_\cG\,\eta}\;d\overline{\eta}_\cI\, d\eta_\cI=
\det\cT_\cG \cdot\e^{-\;\overline{\eta}_\cE\cdot\cK_\cG\,\eta_\cE}
\end{equation}
and thus also $\det \cL_\cG=\det\cT_\cI \cdot \det\cK_\cG$ hold. In particular, if in 
addition $\cL_\cG$ is invertible, then $\cK_\cG$ is also 
invertible and the matrix elements of its inverse are those of the corresponding ones for 
$\cL_\cG^{-1}$ itself
\begin{align}\label{gr170}
(\cK_\cG^{-1})\;_{e,v=\partial(e);e^\prime,v^\prime=\partial(e^\prime)}&=
(\cL_\cG^{-1})\;_{e,v=\partial(e);e^\prime,v^\prime=\partial(e^\prime)}.
\end{align}
\end{theorem}
\begin{proof}
In view of \eqref{gr14}, this theorem is a direct consequence of Lemma \ref{lem:a}.
\end{proof}

\subsection{The generalized star product and two vertex graphs}\label{subsec:2}
In this subsection we will show, how in the case of any two vertex graph $\cG_2$ 
the above construction of $\cK_{\cG_2}$ out of $\cL_{\cG_2}$ is 
also obtained from the generalized star product introduced in \cite{KS4,KS5}.  
The vertices are denoted as $v$ and $v^\prime$. Figure \ref{fig:1} serves as an illustration.

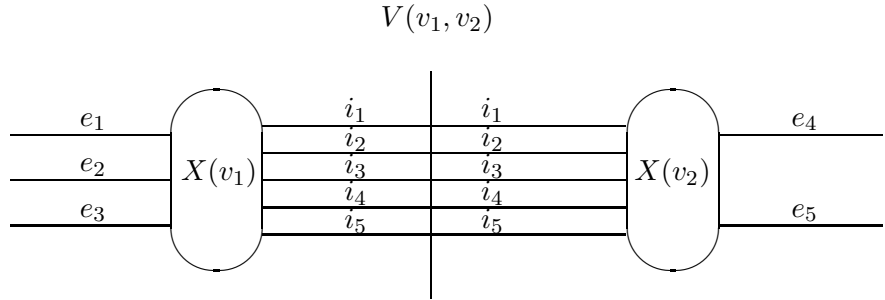
\begin{figure}[htb]
\setlength{\unitlength}{1.2cm}
\begin{picture}(12,6)

\put(3,3){\oval(1.0,2.0)}
\put(8,3){\oval(1.0,2.0)}

\put(3.5,3){\line(1,0){3,98}}
\put(3.5,3.3){\line(1,0){3,98}}
\put(3.5,3.6){\line(1,0){3,98}}
\put(3.5,2.7){\line(1,0){3,98}}
\put(3.5,2.4){\line(1,0){3,98}}

\put(0.74,3.5){\line(1,0){1,75}}
\put(0.74,3){\line(1,0){1,75}}
\put(0.74,2.5){\line(1,0){1,75}}

\put(8.5,3.5){\line(1,0){1,8}}
\put(8.5,2.5){\line(1,0){1,8}}

\put(5.35,4.2){\line(0,-1){2,5}}

\put(4.8,4.7){$V(v_1,v_2)$}
\put(2.61,3.0){$X(v_1)$}
\put(7.58,3.0){$X(v_2)$}


\put(4.4,3.7){$i_1$}
\put(5.9,3.7){$i_1$}
\put(4.4,3.38){$i_2$}
\put(5.9,3.38){$i_2$}
\put(4.4,3.08){$i_3$}
\put(5.9,3.08){$i_3$}
\put(4.4,2.78){$i_4$}
\put(5.9,2.78){$i_4$}
\put(4.4,2.47){$i_5$}
\put(5.9,2.47){$i_5$}

\put(1.5,3.6){$e_1$}
\put(1.5,3.1){$e_2$}
\put(1.5,2.6){$e_3$}

\put(9.3,3.6){$e_4$}
\put(9.3,2.6){$e_5$}
\end{picture}
\caption{\label{fig:1}   Pictorial description of a two vertex graph $\cG_2$ 
appearing in the definition of the generalized star product. $X(v_1)$ is  
an $8\times 8$ matrix, $X(v_2)$ a $7\times 7$ matrix and $V(v_1,v_2)$ a $5\times 5$ matrix. 
The graph has 5 external edges $e_1,\cdots e_5$ and 5 internal edges $i_1,\cdots i_5$. 
The index orderings of $\cE\cup\cI$ and hence of $\cI$ are 
$$
\hspace{-2.5cm}
\cE\cup \cI=\{e_1,\cdots,e_5, i_1,\cdots ,i_5\},\;\cI=\{i_1,\cdots,i_5\}=\cI(v_1)\cap \cI(v_2)
$$
with 
$\cI(v_1)=\{e_1,e_2,e_3,i_1,\cdots,i_5\},\;\cI(v_2)=\{i_1,\cdots,i_5,e_4,e_5\}$.
$\cL_{\cG_2}$ is a $15\times 15$ matrix and $\cL_\cE$ a $5\times 5$ matrix.}
\end{figure}
Let the matrix $X(v_1)$, indexed by $\cI(v_1)$, be given in a $2\times 2$ block form 
\begin{equation}\label{gr18}
X(v_1)=\begin{pmatrix}A&B\\C&D\end{pmatrix}
\end{equation}
where $A$ is an $n_1\times n_1$ matrix, $B$ an $n_1\times p$ matrix, $C$ an $p\times n_1$ matrix 
and finally $D$ a $p\times p$ matrix. Here $n_1=\mid\cE\cap \cI(v_1)\mid$ is the number of external 
edges $e$ terminating at $v_1$, that is $\partial(e)=v_1$. 
$p=n(\cI)$ the number of internal lines such that $n_1+p=n(v_1)$. 
Thus for example $A$ is indexed by $\cE\cap \cI(v_1)$ while $D$ is indexed by $\cI\cap\cI(v)$.
Similarly, write $X(v_2)$, indexed by $\cI(v_2)=(\cI(v_1))$, in 
a $2\times 2$ block matrix form 
\begin{equation}\label{gr19}
X(v_2)=\begin{pmatrix}E&F\\G&H\end{pmatrix}
\end{equation}
where $H$ is a $p\times p$ matrix, $G$ a $p\times m_1$ matrix etc. 
Here $m_1=\mid\cE\cap \cI(v_2)\mid$ is the number of external edges $e$ terminating at $v^\prime$, 
that is $\partial(e)=v_2$. Thus e.g. $E$ is indexed by $\cE\cap \cI(v_2)$. 
Then the $(n_1+p+m_1+p)\times (n_1+p+m_1+p)$ matrix $\cL_{\cG_2}$, which is obtained from 
$X(v_1),X(v_2)$ and $V=V(v_1,v_2)$, takes the form 
\begin{align}\label{gr20}
\cL_\cG&=\begin{pmatrix}\;\;A&\;\,B&\;\;0&\;0\\\;C&\;\,D&\;\;0&\;-V^{-1}\\
                        \;0&\;\,0&\;\;E&\;F\\\;0&\;\,-V&\;\;G&\;H\;
\end{pmatrix}\\\nonumber
&\quad\quad\begin{array}{cccc}\underbrace{\:}_{n_1}&\underbrace{\,}_p
&\underbrace{\,}_{m_1}&\underbrace{\,}_{p}\end{array}
\end{align}
which in a canonical way is indexed by the set $(\cE\cup\cI)\triangleright\cV$.
This gives  
\begin{align}\label{gr21}
\cL_\cE&=\begin{pmatrix}A&0\\0&E\end{pmatrix},\qquad
\cL_{\cE,\cI}=\begin{pmatrix}B&0\\0&F\end{pmatrix}\\\nonumber
\cL_{\cI,\cE}&=\begin{pmatrix}C&0\\0&G\end{pmatrix},\qquad 
\cT_{\cG_2}=\cL_\cI=\begin{pmatrix}D&-V^{-1}\\-V&H\end{pmatrix}
\end{align}
where $\cL_\cE$ is a $\mid \cE\mid \times \mid \cE\mid$ matrix, etc.
Therefore by \eqref{gr16}
\begin{align}\label{gr22}
\cK_{\cG_2}&=\begin{pmatrix}A&0\\0&E\end{pmatrix}-\begin{pmatrix}B&0\\0&F\end{pmatrix}
\begin{pmatrix}D&-V^{-1}\\-V&H\end{pmatrix}^{-1}\begin{pmatrix}C&0\\0&G\end{pmatrix}\\
&=\begin{pmatrix}A-B(D-V^{-1}H^{-1}V)^{-1}C&-B(D-V^{-1}H^{-1}V)^{-1}V^{-1}H^{-1}G\\
-F(H-VD^{-1}V^{-1})^{-1}VD^{-1}C&E-F(H-VD^{-1}V^{-1})^{-1}G\end{pmatrix}.\nonumber
\end{align}
We have used \eqref{schur} by which 
\begin{align*}
\begin{pmatrix}D&-V^{-1}\\-V&H\end{pmatrix}^{-1}
&=\begin{pmatrix}(D-V^{-1}H^{-1}V)^{-1}&(D-V^{-1}H^{-1}V)^{-1}V^{-1}H^{-1}\\
(H-VD^{-1}V^{-1})^{-1}VD^{-1}&(H-VD^{-1}V^{-1})^{-1}\end{pmatrix}
\end{align*}
and we have assumed the matrix $\cT_{\cG_2}=\cL_\cI$ to be invertible.
Alternatively, by Theorem \ref{grassmann:theo}, if $\cL_{\cG_2}^{-1}$ is written in 
$4\times 4$ block form like $\cL_{\cG_2}$  
$$
\cL_{\cG_2}^{-1}=\begin{pmatrix}\alpha&\cdot&\beta&\cdot\\
\cdot&\cdot&\cdot&\cdot\\\gamma&\cdot&\delta&\cdot\\\cdot&\cdot&\cdot&\cdot\end{pmatrix}
$$
then 
\begin{equation}\label{gr221}
\cK_{\cG_2}^{-1}=\begin{pmatrix}\alpha&\beta\\\gamma&\delta\end{pmatrix}
\end{equation}
holds. This can be checked by calculating the inverse of $\cL_{\cG_2}$, a tedious but straightforward 
calculation using \eqref{schur} iteratively.
Moreover we have 
\begin{lemma}\label{lem:1} Assume $\cT_{\cG_2}=\cL_\cI$ to be invertible, such that 
$\cK_{\cG_2}$ is well defined. Then the generalized star product $X(v)\star_{V(v,v^\prime)}X(v^\prime)$ 
as introduced in \cite{KS4} is also well defined and both these quantities are equal. 
\end{lemma}
\begin{proof}
In the present notation, where we recall $V=V(v,v^\prime)$,
\begin{equation}\label{gr23} 
X(v)\star_{V(v,v^\prime)}X(v^\prime)=\left(\begin{array}{cc}A+BK_2HVC&BK_2G\\FK_1C&E+FK_1DV^{-1}G
\end{array}\right)
\end{equation}
holds with 
\begin{equation}\label{gr24}
\begin{aligned}
K_1&=(\1-VDV^{-1}H)^{-1}V\quad=V(\1-DV^{-1}HV)^{-1}\\
K_2&=(\1-V^{-1}HVD)^{-1}V^{-1}=V^{-1}(\1-HVDV^{-1})^{-1},
\end{aligned}
\end{equation}
see Section 4 in \cite{KS4} and Section 3 in \cite{KS5}. We use the relations
$$
\begin{aligned}
(D-V^{-1}H^{-1}V)^{-1}&=-(\1-V^{-1}HVD)^{-1}V^{-1}HV=-K_2HV\\
(H-VD^{-1}V^{-1})^{-1}&=-(\1-VDV^{-1}H)^{-1}VDV^{-1}=-K_1DV^{-1},
\end{aligned}
$$
insert this into \eqref{gr22}. Comparison with \eqref{gr23} gives the claim.
\end{proof}

\subsection{The generalized star product and arbitrary graphs}\label{subsec:3}
We are now able to extend this comparison to the case where the graph has more than two vertices.
The idea for this alternative is 
to carry out the integrations in \eqref{gr17} in steps while using iteratively the 
Grassmannian-Gaussian construction of the star product as given in the previous subsection. This proof will 
give a more explicit representation of $\cK_\cG$ and $\cT_G$ in terms of the original data 
$\cG,\,\underline{X}_\cG$ and $\underline{V}_\cG$. It is important to recall that the data uniquely 
fix $\cL_\cG$. So this alternative proof will be by induction on the number of vertices, by which we 
will construct a sequence of connected graphs $\cG_l$ with $l$ vertices such that 
$\cG_{l=n(\cV)}=\cG$. Similarly we will provide inductively $\cK_{l}$ -- with $\cK_{\cG_{l=n(\cV)}}=\cK_\cG$ 
-- and $\cT_{_l}$, the last one will be given recursively in the form 
\begin{equation}\label{trecurs}
\cT_{l}=\cT_{l-1}\oplus \cT^l
\end{equation}
with suitable $\cT^l$ and where by definition $\cT^1=1$ . 
Here and in what follows we make the notational convention that for any two square 
matrices $M_1$ and $M_2$ their direct sum $M_1\oplus M_2$ is identified with the $2\times 2$ block matrix
$$
\left(\begin{array}{cc}M_1&0\\0&M_2\end{array}\right).
$$
We first construct the $\cG_l$ inductively. As for the case $l=1$, choose any vertex and call it $v_1$. 
Let $\cG_1$ denote a star graph with $n(v_1)$ external lines labeled by  $\cI(v_1)$ as for the graph $\cG$ 
itself. Assume that we have constructed the connected graph $\cG_l$ with the set of 
vertices $\cV_l=\{v_1,\cdots, v_l\}$, named like those of $\cG$. 
Also the edges $i\in\cG_l$ with $v_k\in\partial(i)$ are in one to one 
correspondence with the edges in $\cG$ having $v_k$ in their boundary. Thus we may use $\cI(v_k)$ to 
index these edges. 
Furthermore the set of external and internal edges in $\cG_l$ are such that 
\begin{align}\label{gr241}
\cE_l\cup\cI_l&=\cup_{1\le k\le l}\cI(v_k)\\
\cI_l&=\cup_{1\le k\neq k^\prime\le l}\,\left(\cI(v_k)\cap\cI(v_{k^\prime})\right).\nonumber 
\end{align} 
To obtain $\cG_{l+1}$ from $\cG_l$, observe there is a vertex in $\cG$, denoted $v_{l+1}$, such that  
$$
\bar{n}(v_{l+1})=\sum_{k=1}^l n(v_k,v_{l+1})>0.
$$
$\cG_{l+1}$ is obtained as follows. Take $\cG_{l}$ and a single vertex graph with vertex 
denoted by $v_{l+1}$ and with  $\bar{n}(v_{l+1})$ edges emanating.  
Call this graph $\bar{\cG}(v_{l+1})$. Label its edges by the elements in 
$\cup_{k\le l}\left(\cI(v_k)\cap \cI(v_l)\right)$.
Glue each edge $i\in\cI(v_k)\cap \cI(v_l)$ in $\bar{\cG}(v_{l+1})$ to the edge with the same index in 
$\cG_l$. In case $\cG$ is a metric graph with set of internal edge lengths $\underline{a}$, give the 
resulting internal edge $i$ in $\cG_{l+1}$ the length $a_i$.
To sum up,  the resulting graph $\cG_{l+1}$ has $n(v_k,v_{l+1})$ edges connecting $v_k$ with $v_{l+1}$ 
in $\cG_{l+1}$. In total $\cG_{l+1}$ thus obtained has edges which also are of the form \eqref{gr241} 
with $l$ being replaced by $l+1$.
This concludes the inductive construction of the graphs and gives $\cG$ as $\cG_{l=n(\cV)}$.

$\cV_l$ can be viewed as a subset of $\cV$ and that by \eqref{gr241} $\cE_l\cup\cI_l$ 
can be viewed as a subset of both $\cE_{l+1}\cup\cI_{l+1}$ 
and $\cE\cup\cI$. Similarly $\cI_l$ can be viewed as a subset of both $\cI_{l+1}$ and $\cI$.
For the $\cE_l$ similar relations, however, are not valid, since an edge in $\cE_{l-1}$ can turn into 
an edge in $\cI_{l}$. More explicitly we introduce the sets 
\begin{equation}\label{gr2411}
\bar{\cI}_l=\cE_{l-1}\cap \cI_{l}=\cup_{k:k<l}\left(\cI(v_k)\cap\cI(v_l)\right)
\end{equation}
which will be used in the sequel. They satisfy
\begin{equation}\label{gr2412}
\bar{\cI}_l\cap\bar{\cI}_{l^\prime}=\emptyset,\quad l\neq l^\prime;\qquad 
\cup_{1\le l\le n(\cV)}\bar{\cI}_l=\cI.
\end{equation}
Pictorially this construction can be understood as follows. $\cG_l$ is obtained from $\cG$ 
by cutting any internal edge, which connects any vertex 
$v_k\, (1\le k\le l)$ to any vertex $v$ different from the $v_{k^\prime}\, (1\le k^\prime\le l)$. Any such 
edge is then replaced by an infinite half-line. 
As a consequence of this discussion
\begin{equation}\label{gr2413}
\bar{\cI}_l\triangleright\cV_l\subseteq 
(\cE_l\cup\cI_l)\triangleright \cV_l\subset (\cE_{l+1}\cup\cI_{l+1})\triangleright \cV_{l+1}
\subseteq (\cE\cup\cI)\triangleright \cV,\qquad 1\le l\le n(\cV)-1,
\end{equation}
which induces an ordering on these sets.
In order to construct the $\cK_{l}$ we introduce the matrices 
\begin{equation}\label{gr242}
\overline{V}(v_{l})=V(v_1,v_{l})\oplus V(v_2,v_{l})\cdots\oplus V(v_{l-1},v_{l}),\quad l\ge 2 
\end{equation}
which are invertible. Thus in the example given by Figure \ref{fig:2} $\overline{V}(v_{l=3})$ is a 
$3\times 3$ matrix.

Set $\cK_{1}=X(v_1)$ and inductively 
\begin{equation}\label{gr243}
\cK_{l+1}=\cK_{l}\star_{\:\overline{V}(v_{l+1})} X(v_{l+1}). 
\end{equation}
In particular $\cK_{2}$ is just as given by Lemma \ref{lem:1}. We note that the 
invertibility of a certain matrix is necessary, see the discussion of $\cT_{\cG_2}$ in Subsection 
\ref{subsec:2}. So if the invertibility of certain matrices holds -- see also below -- 
the associativity of the generalized star product \cite{KS4} gives
\begin{equation}\label{gr25}
\cK_{l}=X(v_{1})\star_{V(v_1,v_2)}X(v_{2})\star_{\overline{V}(v_{3})}X(v_{3})
\cdots\star_{\overline{V}(v_{n(l)})}X(v_{l}). 
\end{equation}
We now repeat this construction by performing Grassmann integration over 
$$
\e^{-\overline{\eta}\cdot \cL_{\cG}\eta}
$$
in steps. In order to do this we  
view $\cG_{l}$ as a single vertex graph with a vertex denoted by $\bar{v}_l$ and 
with edges labeled by $\cE_l$. Combine it with
the single vertex graph $\bar{\cG}(v_{l+1})$ and with connecting matrix given as 
$V(\bar{v}_l,v_{l+1})=\overline{V}(v_{l+1})$. Correspondingly we take as data for $\cG_{l+1}$ 
\begin{equation}\label{gr26}
\underline{X}_{l+1}=\{\cK_{l}, X(v_{l+1})\},\quad \underline{V}_{l+1}=
\{V(\bar{v}_l,v_{l+1})=\overline{V}(v_{l+1})\}
\end{equation}
and out of this we form the fermionic Lagrangean
\begin{equation}\label{gr27}
\overline{\eta}\cdot \cL_{l+1}\eta=
\overline{\eta}\cdot \cK_{\cG_l}\eta+\overline{\eta}\cdot X(v_{l+1})\eta
-\overline{\eta}\cdot\overline{V}(v_{l+1})\eta.
\end{equation}
In order not to burden the notation, we have not stated explicitly, which Grassmann variables  
out of the set \eqref{grassmannset} are involved. Indeed, which ones are involved can be read off the 
index set associated to the matrices $\cK_{l},X(v_{l+1})$ and $\overline{V}(v_{l+1})$. 

With this notational convention and by Lemma \ref{lem:1} we obtain 
\begin{equation}\label{gr28}
\det \cT^{l+1}\cdot \e^{-\overline{\eta}\cdot \cK_{l+1}\eta}=\int \e^{-\overline{\eta}\cdot \cL_{l+1}\eta}
d\overline{\eta}_{\bar{I}_{l+1}}d\eta_{\bar{I}_{l+1}}.
\end{equation}
where $\cT^{l+1}=(\cL_{l+1})_{\bar{\cI}_{l+1}}$ -- in an adaption of the notation used in \eqref{gr14} --
is invertible. We recall that by definition $\cT^1=1$. This concludes the recursive construction of 
$\cG_l,\cK_{l}$ and 
\begin{equation}\label{gr29}
\cT_{l}=\oplus_{k=1}^l\cT^k.
\end{equation}

We iterate the recursion \eqref{gr28} in combination with recursion \eqref{gr27}, use \eqref{trecurs}, 
\eqref{gr2412} and \eqref{gr25} to obtain
\begin{equation}\label{gr30}
\det \cT_{l=n(\cV)}\cdot \e^{-\overline{\eta}\cdot \cK_{l=n(\cV)}\eta}=
\int \e^{-\overline{\eta}\cdot \cL_{\cG}\eta}d\overline{\eta}_{\cI}d\eta_{\cI}.
\end{equation}
Comparison with \eqref{gr17}, while keeping Remark \ref{re:1} in mind, gives the main result of this 
article.
\begin{theorem}\label{theo:main} Assume that the matrices $\cT_\cG$ and $\cT^l$ are all invertible, 
such that all $\cT_{l}$ are also invertible. 
Then the quantities $\cK_\cG$ and $\det T_\cG$ as given by  \eqref{gr17} are equal to 
$\cK_{l=n(\cV)}$ and $\det \cT_{l=n(\cV)}$ respectively, 
where $\cK_{l=n(\cV)}$ and $\cT_{l=n(\cV)}$ are defined by \eqref{gr25} and \eqref{gr29}.
\end{theorem}
Without proof we state that under the assumptions stated actually $\cT_\cG\simeq \cT_{l=n(\cV)}$ holds.

\begin{corollary}
For a given graph $\cG$ let the data $\underline{X}_\cG$ and $\underline{V}_\cG$ consist of unitaries. 
Under the corresponding invertibility assumptions, $\cK_{\cG}$ is also unitary as are in fact all 
$\cK_{\cG_l}$.
\end{corollary}
\begin{proof}
This follows from the theorem and the results in \cite{KS4,KS5}.
\end{proof} 
\begin{remark}\label{re:Srep} 
In view of Remark \ref{re:scattering}, this corollary is relevant in the context of 
quantum scattering theory on graphs. Then $\cK_{\cG}$ is the scattering matrix at a fixed energy $E$
associated with the entire metric graph $\cG$ and where the metric enters through the connecting 
matrices given in the form \eqref{metricconn}.
Relation (3.36) in \cite{KS6} provides another way to obtain this scattering matrix in 
terms of the single vertex scattering matrices. Also in \cite{KS6} a series expansion of every matrix 
element is given. This expansion is indexed by so called walks ${\bf w}$ with 
length $\mid \bf w\mid$ 
and has expansion coefficients 
of the form $\exp(\ii \sqrt{E}|{\bf w}|)$ times a monomial in the single vertex scattering 
matrix elements. It has been used to formulate a new approach to the traveling salesman problem.
\end{remark}


\subsection*{Acknowledgments:} The authors would like to thank M. Karowski for valuable discussions.

\vspace{1cm}

\begin{appendix}
\section{The generalized star product and Gaussian integrals}\label{propagator}
\renewcommand{\theequation}{\mbox{\Alph{section}.\arabic{equation}}}
\setcounter{equation}{0}


In this appendix we give a \emph{bosonic} discussion using Gaussian distributions.
We start by recalling some basic facts about Gaussian 
distributions, also in order establish notation and for the sake of comparison with our \emph{fermionic}
discussion.
Let $x=(x_1,x_2,\cdots,x_n)$ and $y=(y_1,y_2,\cdots,y_n)$ denote elements in $\R^n$ and set 
$$
y\cdot x=\sum_{i=1}^n y_ix_i=x\cdot y
$$
and for integration $dx=\prod_{i=1}^ndx_i$ denotes the infinitesimal volume element on $\R^n$.
Let $A$ be a real symmetric matrix, which in addition is \emph{positive (definite)}
$$
x\cdot Ax=\sum_{i,j=1}^n x_iX_{ij}x_j>0,\quad x\neq 0
$$
and then we write $A>0$.
Then also $\det A>0$ and in addition $A^{-1}$ exists, is real, symmetric and positive. 
$\1$ denotes the unit matrix in the given context and $A>A^\prime$ means 
$x\cdot Ax>x\cdot A^\prime x$ for all $x\neq 0$. If $\kappa\1<A<\mu\1$, 
then $\mu^{-1}\1<A^{-1}<\kappa^{-1}\1$.

We make the following notational convention. Here and in what follows, whenever $x\in \R^n$ 
stands to the right of an $m\times n$ matrix $A$, then it is viewed 
as a column vector. The outcome $Ax$ will also be interpreted as a column vector in $\R^m$. 
When $x$ stands to the left of $A$ it will be viewed as a row vector.

Define the Gauss distribution with covariance $A>0$ via its probability measure 
\begin{equation}\label{gauss1}
d\mu_A(x)=\frac{(\det A)^{1/2}}{(2\pi)^{n/2}}\e^{-\frac{1}{2}x\cdot Ax}dx
\end{equation}
on $\R^n$. That this is a probability measure follows from
\begin{equation}\label{gauss02}
\int_{x\in\R^n}\e^{-\frac{1}{2}x\cdot Ax}dx=\frac{(2\pi)^{n/2}}{(\det A)^{1/2}}
\end{equation}
and is the analogue of \eqref{gr7}.

The first two moments of the measure $d\mu_A(x)$ are
\begin{equation}\label{gauss6}
\begin{aligned}
\int_{\R^n}x_id\mu_A(x)&=0\\
\int_{\R^n}x_ix_jd\mu_A(x)&=A^{-1}_{ij}
\end{aligned}
\end{equation}
which are the analogues of \eqref{gr11}.

Write any $x\in\R^{n}$ as $x=(x^{(1)},x^{(2)})$ with 
$x^{(1)}=(x_1,\cdots, x_p)\in\R^{p},x^{(2)}=(x_{p+1},\cdots,x_{n})\in\R^{n-p}$ and set
\begin{align*}
x^{(1)}\cdot A_{11}x^{(1)}&=\sum_{i,j\le p}x_iA_{ij}x_j,\qquad x^{(1)}\cdot A_{11}x^{(2) }
=\sum_{i\le p,p+1\le j}x_iA_{ij}x_j\\
x^{(2)}\cdot A_{21}x^{(1)}&=\sum_{p+1\le i,j\le p}x_iA_{ij}x_j,
\qquad x^{(2)}\cdot A_{22}x^{(2)}=\sum_{p+1\le i,j\le p}x_iA_{ij}x_j
\end{align*}
such that we have the decomposition
\begin{align*}
 x\cdot Ax=x^{(1)}\cdot A_{11}x^{(1)}+x^{(1)}\cdot A_{12}x^{(2)}+x^{(2)}\cdot A_{21}x^{(1)}
+x^{(2)}\cdot A_{22}x^{(2)}.
\end{align*}
In other words, we use the $2\times 2$ block decomposition \eqref{Ablock} of the matrix $A$.
Since $A$ is assumed to be positive definite, so are $A_{11}$ and $A_{22}$ and their inverses. Also 
$A_{21}$ is the transpose of $A_{12}$.
So $\widehat{A}=A_{11}-A_{12}A_{22}^{-1}A_{21}$ is a well defined and symmetric $p\times p$ matrix.
In fact it is positive definite, see, e.g., \cite{Zhang}. Actually we need a stronger result. 
\begin{lemma}\label{lem:4}
If $A>\kappa\1$ with $\kappa>0$ holds, then also $\widehat{A}>\kappa\1$ is valid for the Schur 
complement of $A_{22}$.
\end{lemma}
\begin{proof} Under the assumption $0<A^{-1}<\kappa^{-1}\1$, hence also $0<\widehat{A}^{-1}<\kappa^{-1}\1$
due to \eqref{schur}. Taking the inverse gives $\widehat{A}>\kappa\1$.\end{proof}
We leave the proof of the following lemma to the reader. It is the analogue of Lemma \ref{lem:a}.
\begin{lemma}\label{lem:5}
The following relation holds
\begin{equation}\label{gauss61}
\int_{x^{(2)}\in\R^{n-p}}\e^{-\frac{1}{2}x\cdot Ax}\,dx^{(2)}=
\frac{(2\pi)^{n-p/2}}{\left(\det A_{22}\right)^{1/2}}\e^{-\frac{1}{2}x^{(1)}\cdot \widehat{A}x^{(1)}}.
\end{equation}
\end{lemma}
Since $\widehat{A}$ is positive definite, we may integrate \eqref{gauss61} 
over $x^{(1)}$ and $\det A=\det A_{22}\cdot \det \widehat{A}$ follows, which is relation \eqref{gr113},
however, in the restricted context of positive $A$.

We turn to the generalized star product and a way to obtain it through Gaussian integrals.
Consider $X(v_1)$ and $X(v_2)$ of the form \eqref{gr18} and \eqref{gr19}. We use the notation 
employed in this context.       
\begin{proposition}\label{prop:app}
Assume $X(v_1)>\kappa\1,X(v_2)>\kappa\1$ with $\kappa>1$ and let $V(v_1,v_2)$ be an \emph{orthogonal} 
$p\times p$ matrix. Then $X(v_1)\star_{V(v_1,v_2)}X(v_2)>(\kappa-1)\1$.
\end{proposition}
\begin{proof}
Let $\cL_{\cG_2}$ be as in \eqref{gr20} but now indexed from $1$ to $n_1+p+m_1+p$. Also let the orthogonal 
matrix $V=V(v_1,v_2)=V(v_2,v_1)^{-1}=V(v_2,v_1)^{T}$ be 
indexed from $1$ to $p$.
For $0\neq x\in\R^{n_1+m_1+2p}$ by Schwarz inequality 
\begin{align*}
x\cdot \cL_{\cG_2} x&=\sum_{1\le i,j\le n_1+p}x_iX(v_1)_{ij}\,x_j
+\sum_{n_1+p+1\le i,j\le n_1+m_1+2p}x_iX(v_1)_{ij}\,x_j\\
&\qquad-\sum_{1\le i,j\le p}x_{n_1+i}V_{ij}\,x_{n_1+p+m_1+j}
-\sum_{1\le i,j\le p}x_{n_1+p+m_1+i}V_{ji}\,x_{n_1+j}\\
&>(\kappa -1)x\cdot x. 
\end{align*}
But $X(v_1)\star_{V(v_1,v_2)}X(v_2)$ is a Schur complement by the discussion in Subsection \ref{subsec:2},
that is $X(v_1)\star_{V(v_1,v_2)}X(v_2)=\cK_{\cG_2}$, and so the claim follows from Lemma \ref{lem:4}.
\end{proof}
Write $x=(z^{(1)},z^{(2)})\in\R^{n_1+m_1+2p}$ with 
$z^{(1)}=(x_1,\cdots,x_{n_1},x_{n_1+2p+1},\cdots,x_{n_1+2p+m_1})\in\R^{n_1+m_1}$
and $z^{(2)}=(x_{n_1+1},\cdots,x_{n_1+2p})\in\R^{2p}$. Then with the notation used in 
Subsection \ref{subsec:2} and in the proof of the lemma
we obtain
\begin{lemma}\label{lem:a3}
With the assumptions as in Lemma \ref{lem:5} the relation
\begin{equation}\label{gauss7}
\int_{z^{(2)}\in\R^{2p}}\e^{-\frac{1}{2}x\cdot \cL_{\cG_2} x}dz^{(2)}=\frac{(2\pi)^{p}}{(\det \cT_{\cG_2})^{1/2}}
\e^{-\frac{1}{2}z^{(1)}\cdot\cK_{\cG_2} z^{(1)}}
\end{equation}
holds.
\end{lemma}
We turn to an arbitrary graph $\cG$ with data $\underline{X}_\cG$ and $\underline{V}_\cG$ 
with the property that each $X(v)$ is positive definite and each $V(v,v^\prime)$ is orthogonal.
Introduce the variable 
$$
x=\{x_{i,v}\}_{i,v\in (\cE\cup\cI)\triangleright\cV}=\{ x_{i,v}\}_{i,v:i\in\cI(v),v\in\cV}
\in\R^{\mid\cE\mid+2\mid\cI\mid},
$$
let the matrices $\cL(v)$ and $\cL(v,v^\prime)$ be as in \eqref{gr12} and set 
\begin{equation}\label{gauss8}
\begin{aligned}
x\cdot \cL(v)x&=\sum_{i,j\in\cI(v)}x_{i,v}X(v)_{ij}x_{j,v}\\
x\cdot \cL(v,v^\prime)x&
=-\sum_{i,j\in\cI(v)\cap\cI(v^\prime)}x_{i,v}V(v,v^\prime)_{ij}x_{j,v^\prime}.
\end{aligned}
\end{equation}
Define the \emph{quadratic, bosonic Lagrangian} as
\begin{equation}\label{gauss9}
x\cdot \cL_\cG\, x=\sum_{v}x\cdot \cL(v)\,x
-\sum_{v\neq v^\prime}x\cdot \cL(v,v^\prime)\,x.
\end{equation}
We decompose $x$ into its \emph{exterior} and \emph{interior components}, that is $x=(x_\cE,x_\cI)$ with
$$
x_{\cE}=\{x_{e,v=\partial(e)}\}_{e\in\cE},\qquad x_{\cI}=\{x_{i,v}\}_{i\in\cI(v),v\in \cV}.
$$
and correspondingly we get 
\begin{align}\label{gauss10}
x\cdot \cL_\cG\, x&=x_\cE\cdot \cL_\cE\,x_\cE+
x_\cE\cdot \cL_{\cE,\cI}\, x_\cI+x_\cI\cdot \cL_{\cI,\cE}\, x_{\cE}+
x_\cI\cdot \cL_\cI\, x_\cI\\\nonumber
&=(x_\cI+\left(\cL_\cI\right)^{-1\,T}\left(\cL_{\cE,\cI}\right)^T\,x_\cE)\cdot \cL_\cI\,
(x_\cI+\left(\cL_\cI\right)^{-1}\cL_{\cI,\cE}\, x_\cE)\\\nonumber
&\qquad \qquad+x_\cE\cdot
\left(\cL_\cE-\cL_{\cE,\cI}\left(\cL_\cI\right)^{-1}
\cL_{\cI,\cE}\right)x_\cE.\nonumber
\end{align}
The following extension of Lemma \ref{lem:a3} to arbitrary graphs is valid
\begin{proposition}\label{prop:a4}
Given data $\underline{X}_\cG=\{X(v)\}_{v\in\cG}$ and 
$\underline{V}_\cG=\{V(v,v^\prime)\}_{v\neq v^\prime\in \cV}$ with 
$X(v)>(n(\cV)-1)\1$ and orthogonal $V(v,v^\prime)$, the $\mid\cE\mid\times \mid \cE\mid$ 
matrix $\cK_\cG$ defined by 
\begin{equation}\label{gauss11}
\int_{x_\cI\in\R^{2\mid\cI\mid}}\e^{-\frac{1}{2}x\cdot \cL_{\cG} x}dx_\cI=
\frac{(2\pi)^{\mid\cI\mid}}{(\det \cT_{\cG})^{1/2}}
\e^{-\frac{1}{2}x_\cE\cdot\cK_{\cG} x_\cE}
\end{equation}
is positive definite. 
\end{proposition}
\begin{proof}
We use the representation \eqref{gr25} and Lemma \ref{lem:4} repeatedly. Observe that 
$\overline{V}(v_{l+1})$ defined by \eqref{gr242} is also orthogonal.
\end{proof}
\end{appendix}


\end{document}